\documentclass[12pt,a4paper]{article}
\usepackage[utf8]{inputenc}
\usepackage[T1]{fontenc}
\usepackage{amsmath}
\usepackage{amsfonts}
\usepackage{amssymb}
\usepackage{makeidx}
\usepackage{indentfirst}
\usepackage[left=3cm,right=3cm,top=4cm,bottom=3cm]{geometry}

\newcommand{\eqdef}{\triangleq}

\usepackage{tkz-graph}
\usetikzlibrary{plotmarks,trees,arrows}
\usetikzlibrary{graphs,graphs.standard} 

\usepackage{graphicx,color}

\usepackage{graphicx}
\usepackage[all]{xy}

\newtheorem{theorem}{Theorem}

\newtheorem{corollary}[theorem]{Corollary}

\newtheorem{definition}[theorem]{Definition}
\newtheorem{example}[theorem]{Example}

\newtheorem{proposition}[theorem]{Proposition}
\newtheorem{remark}[theorem]{Remark}

\newenvironment{proof}[1][Proof]{\noindent\textbf{#1.} }{\ \rule{0.5em}{0.5em}}

\date{}
\author{\large Luciano Panek and Nayene Michele Pai\~ao Panek
\thanks{Luciano Panek is with the Center of Exact Sciences and Engineering, State University of West Paran\'a, Foz do Igua\c{c}u, Paran\'a, Brazil (e-mail: luciano.panek@unioeste.br).}
\thanks{Nayene Michele Pai\~ao Panek is with the Center of Exact Sciences and Engineering, State University of West Paran\'a, Foz do Igua\c{c}u, Paran\'a, Brazil (e-mail: nayene.panek@unioeste.br).}
\thanks{Copyright (c) 2021 IEEE. Personal use of this material is permitted.  However, permission to use this material for any other purposes must be obtained from the IEEE by sending a request to pubs-permissions@ieee.org.}
}

\title{Optimal Anticodes, Diameter Perfect Codes, Chains and Weights}

\begin{document}

\maketitle

\begin{abstract}
Let $P$ be a partial order on $[n] = \{1,2,\ldots,n\}$,  $\mathbb{F}_{q}^n$ be the linear space of $n$-tuples over a finite field $\mathbb{F}_{q}$ and  $w$ be a weight on $\mathbb{F}_{q}$. In this paper, we consider metrics on $\mathbb{F}_{q}^n$ induced by chain orders $P$ over $[n]$ and weights $w$ over $\mathbb{F}_q$, and we determine the cardinality of all optimal anticodes and completely classify them. Moreover, we determine all diameter perfect codes for a set of relevant instances on the aforementioned metric spaces.
		
\vspace{0.5cm} 
		
\end{abstract}

	
\section{Introduction}

Classically, coding theory takes place in linear spaces
over finite fields or modules over rings, endowed with a metric, e.g., the linear space $\mathbb{F}_q^n$ of all $n$-tuples over a finite field $\mathbb{F}_{q}$ endowed with the Hamming metric or the module $\mathbb{Z}_{m}^n$ of all $n$-tuples over a ring $\mathbb{Z}_{m}$ of integers modulo $m$ endowed with the Lee metric. 

In a given metric space, codes that attain the sphere-packing bound are called perfect. A possible general setting for the existence problem of perfect codes is the class of distance regular graphs, introduced by Biggs (see \cite{Biggs}), including the nearly ubiquitous Hamming metric spaces (also called Hamming graphs).

The Johnson graphs and the Grassmann graphs are other examples of distance regular graphs (see \cite{dist_reg_graphs}). For the Hamming graphs over $\mathbb{F}_q$, there are no nontrivial perfect codes except for the codes having the parameters of the Hamming codes and the two Golay codes. Martin and Zhu \cite{Martin_Grassman_perfect_codes} showed that there are no nontrivial perfect codes in the Grassmann graphs. Determining all perfect codes is an open problem for Johnson graphs. It was conjectured by Delsarte in the 1970's that there are no nontrivial perfect codes in Johnson graphs (see \cite{Delsarte}). See \cite{Tuvi_Deslsarte_conjecture} and the references therein for progress towards proving the conjecture of Delsarte.

In his pioneer work \cite{Delsarte}, Delsarte also proved the following result:

\begin{theorem}[Delsarte]\label{theo_Delsarte}
	Let $\Gamma = (V,E)$ be a distance regular graph. Let $X$ and $Y$ be subsets of $V$ such that the nonzero distances occurring between vertices of $X$ do not occur between vertices of $Y$. Then $$|X| \cdot |Y| \leq |V|.$$
\end{theorem}

Ahlswede, Aydinian, and Khachatrian in \cite{ahl} gave the definition of diameter perfect code. They examined a variant of Theorem \ref{theo_Delsarte}. Let $\Gamma = (V,E)$ be a distance regular graph. A subset $A$ of $V$ is called an \emph{anticode} with \emph{diameter} $\delta$ if $\delta$ is the maximum graph distance occurring between vertices of $A$. Anticodes with diameter $\delta$ having maximal size are called \emph{optimal anticodes}. Let $diam(A)$ be the diameter of an anticode $A$ in $\Gamma$. Now let $$A^*(D) \eqdef \max\{|A| : diam(A) \leq D\}.$$

\begin{theorem}
	Let $\Gamma = (V,E)$ be a distance regular graph. If $C$ is a code in $\Gamma$ with minimum distance $D + 1$, then
\begin{equation}\label{bound_theo_Ahl}
	A^*(D) \cdot |C| \leq |V|.
\end{equation}
\end{theorem}

Ahlswede \emph{et al.} continued with the following new definition. A code $C$ with minimum distance $D+1$ is called \emph{diameter perfect} if equality holds in equation (\ref{bound_theo_Ahl}). This is a generalization of the usual definition of $e$-perfect code as balls of radius $e$ are anticodes with diameter $2e$.

In Hamming graphs, in addition to the Hamming and Golay codes, the extended Hamming and extended Golay codes are diameter perfect, as well as all MDS codes. In the Johnson graphs, no nontrivial $e$-perfect codes are known but all Steiner systems are diameter perfect codes. Nontrivial diameter perfect codes are also known in the Grassmann graphs. For more details, see \cite{ahl}.

Another possible general setting for the existence problem of perfect  and diameter perfect codes is the class of the metric spaces introduced by Panek and Pinheiro in \cite{Panek_Pinheiro}, the so-called \emph{weighted poset metric spaces}, that include any additive metric space (e.g., the Hamming and Lee metric spaces), as well as the poset metric spaces \cite{Bru} and the pomset metric spaces \cite{Sudha}. The class of these spaces is distinct from the class of the distance regular graph metric spaces and, thus, Delsarte's Theorem, does not apply. Fortunately, an alternative proof that was given in \cite{Deza_set_antiset} can be slightly modified to work on this case (Theorem \ref{theo_Delsart_inv_trans}).

The poset metric spaces were introduced by Brualdi, Graves, and Lawrence in \cite{Bru}. These metrics cast new light into many of the
classical invariants of coding theory (such as minimum distance, packing and
covering radius) and many of their basic results (concerning perfect
and MDS codes, MacWilliams' identity, syndrome decoding and so on) with several works published over the years, contributing to a better understanding of these invariants and their properties when considering the classical Hamming metric. As a comprehensive survey, we cite the book of Firer \textit{et al.} \cite{Firer_book}.

For the distance regular graphs, the balls of radius $e$ are anticodes with diameter larger than $e$ and are optimal anticodes with diameter equal to $2e$ if $C$ is a perfect code.  
For the weighted poset metric spaces, the diameter of the balls of radius $e$ can be equal to $e$. The weighted poset metric is a mix of two extremal cases: the Hamming metric (determined by an anti-chain order and the Hamming weight on coordinates)  with the Niederreiter-Rosenbloom-Tsfasman metric, introduced by Niederreiter in \cite{Niede}, and Rosenbloom and Tsfasman in \cite{RT} (determined by a chain order and the Hamming weight on coordinates; an ultrametric). For the Hamming metric, the diameter of a ball of radius $e$ is greater than $e$. For the Niederreiter-Rosenbloom-Tsfasman metric the diameter of a ball of radius $e$ is exactly equal to $e$.

In this work let us consider the particular setting of weighted poset metrics where the poset is a chain order (Section \ref{sec_NRT}). We will show that the diameter of a ball with radius $e$ is equal to $e$ for all $e$ if and only if the poset is a chain and the weight on coordinates is non-archimedean, the case where the weighted poset metric is an ultrametric (Section \ref{section_nonarchimedian}, Theorem \ref{theo_dimB=D}). In addition, we will describe all optimal anticodes (Section \ref{section_optimalanticodes}, Theorem \ref{lemma_diam_S_w-1} and Theorem \ref{theo2_A*}) and determine for some relevant instances all diameter perfect codes (Section \ref{sec_diam_perf_codes}, Corollary \ref{corol_22}, Corollary \ref{corol_24} and Theorem \ref{theo_26}; Section \ref{sec_diam_per_p_prime}, Theorem \ref{theo_p_prime}).
The Theorem \ref{theo_quando_NRT=Delsate} (Section \ref{section_optimalanticodes}) presents conditions on the weight on coordinates for a code to be diameter perfect. Section \ref{sec_weigths_merrics} is an introduction to weights, metrics, and related concepts used throughout this work. 


\section{Weights, Metrics and Related Concepts}\label{sec_weigths_merrics}

\subsection{Weights and Metrics}\label{sec_peso}

Let $\left[  n\right]  \eqdef\{1,2,\ldots,n\}$ be a finite set with $n$ elements
and $\leq_P$ be a partial order on $\left[  n\right]  $. We call the pair
$P=(\left[  n\right]  ,\leq_P)$ a \emph{poset}. If $k\leq_P j$ and $k\neq j$, we say that $k$ is
\emph{smaller than} $j$ and we write $k <_P j$. An \emph{ideal} in
$P = (\left[  n\right]  ,\leq_P)$ is a subset $I\subseteq\left[  n\right]  $ that
contains every element in $[n]$ smaller than or equal to some of its elements,
i.e., if $j\in I$ and $k\leq_P j$ then $k\in I$. Given a subset $X\subset[n]$,
we denote by $\langle X\rangle$ the smallest ideal containing $X$, called the
\emph{ideal generated} by $X$.

Let $\mathbb{F}_{q}^{n}$ be the space of $n$-tuples over the finite field $\mathbb{F}_{q}$. Given a poset $P=\left(  \left[  n\right]  ,\leq_P\right)  $ and
$u=(u_{1},u_{2},\ldots,u_{n}) \in \mathbb{F}_{q}^n$, the \emph{support} of $u$ is the set
\[
supp(u) \eqdef \left\{  i\in\left[  n\right]
:u_{i}\neq0\right\}.
\]
The ideal  $\langle supp(u) \rangle$ of $P$ is denoted by $I_u^P$ and its set of all maximal elements is denoted by $M_u^P$.

A map $w :\mathbb{F}_q^n \rightarrow \mathbb{N}$ is called a \textit{weight} on $\mathbb{F}_q^n$ if it satisfies the following properties:
\begin{enumerate}
		\item $w(u) \geq 0$ for all $u \in \mathbb{F}_q^n$ and $w(u)=0$ if and only if $u=0$;
		\item $w(u)=w(-u)$ for all $u \in \mathbb{F}_q^n$;
		\item $w(u+v) \leq w(u) + w(v)$ hold for all $u,v \in \mathbb{F}_q^n$ (\textit{triangle inequality}).
	\end{enumerate}

It is clear that, given a weight $w$ on $\mathbb{F}_q^n$, if we define the map $d : \mathbb{F}_q^n \times \mathbb{F}_q^n \rightarrow \mathbb{N}$ by
\begin{eqnarray*}\label{wpmetric}
d(u,v) \eqdef w(u-v), 
\end{eqnarray*}
 then $d$ is a metric invariant by translations.

\begin{definition}\label{def_wcpw}
Given a poset $P=\left(  \left[  n\right]  ,\leq_P\right)  $ and a weight $w$ on $\mathbb{F}_{q}$, the $\left(  P,w\right)  $\emph{-weight} of $u \in \mathbb{F}_q^n$ is the non-negative integer
\begin{eqnarray*}\label{wpmetric}
\varpi_{(P,w)}(u) \eqdef \sum_{i \in M_u^P}w(u_i) + M_w \cdot |I_u^P \setminus M_u^P| 
\end{eqnarray*}
where $M_w \eqdef \max\{w(\alpha):\alpha \in \mathbb{F}_q\}$. If $u,v \in \mathbb{F}_{q}^{n}$, then their $\left(
P,w\right)  $\emph{-distance} is defined by
\[
d_{\left(  P,w\right)  }\left(  u,v\right)  \eqdef \varpi_{\left(  P,w\right)
}\left(  u-v\right)  .
\]
The $(P,w)$-weight $\varpi_{(P,w)}$ and the $(P,w)$-distance $d_{(P,w)}$ are also called  \emph{weighted poset weight} and \emph{weighted poset distance}, respectively.
\end{definition}

The $(P,w)$-weight is a weight on $\mathbb{F}_{q}^{n}$
 (see \cite{Panek_Pinheiro}, Proposition 7), and therefore the $\left(  P,w\right)$-distance is a metric invariant by translations, which combines and
extends several classic weights of coding theory. 
When the weight $w$ is the Hamming weight, the
$\left(  P,w\right) $-weight is the poset weight
proposed by Brualdi \textit{et al.} in \cite{Bru},
and when the weight $w$ is the Lee weight, the
$\left(  P,w\right) $-weight is the pomset weight proposed by Sudha and Selvaraj in \cite{Sudha} (see \cite{Panek_Pinheiro}, Proposition 11). When $P$ is the antichain order with 
$n$ elements, i.e., $i\leq_P j$ in $P$ if and only if $i=j$, the
$\left(  P,w\right)  $-weight is an additive weight. The diagram in Figure \ref{diagram_metrics} illustrates these facts.

\begin{remark}
 We stress that the triangle inequality can fail if we consider $supp(u)$ instead of $M_u^P$ in the definition of $\varpi_{\left(  P,w\right)}(u)$: for the usual order $1 \leq_P 2$ on $[2]$ and the Lee weight $w_L$ on $\mathbb{F}_7$, if $u = (1,0), v = (-1,1) \in \mathbb{F}_7^2$, then $\varpi_{\left(  P,w_L\right)
}\left(  u+v\right) \leq \varpi_{\left(  P,w_L\right)
}\left(  u\right) + \varpi_{\left(  P,w_L\right)
}\left(  v\right)$ if and only if $M_{w_L} \leq 2w_L(1)$, which is not true since $M_{w_L} = 3$ and $w_L(1)=1$.
\end{remark}

\begin{figure*}[!t]		\[
		\xymatrix{
				  & *+[F-:<3pt>]{\small\txt{weighted poset weight}} \ar@{->}[d]|-{\tiny\text{anti-chain}} \ar@{->}[dr]|-{\tiny\text{Lee weight}}  &   \\
				  *+[F-:<3pt>]{\small\txt{poset weight}} \ar@{<-}[ur]|-{\tiny\text{Hamming weight}} & *+[F-:<3pt>]{\small\txt{additive weight}}\ar@{->}[dr]\ar@{->}[dl] &  *+[F-:<3pt>]{\small\txt{pomset weight}} \ar@/^1.3cm/[ll]^{\scriptsize\txt{$\mathbb{F}_2$ and $\mathbb{F}_3$ only}}\\
				  *+[F-:<3pt>]{\small\txt{Hamming weight}}\ar@{<-}[u]|-{\tiny\text{anti-chain}}  &  & *+[F-:<3pt>]{\small\txt{Lee weight}}\ar@{<-}[u]|-{\tiny\text{anti-chain}}  \\
				  &  *+[F-:<3pt>]{\small\txt{Hamming and Lee weights \\ on $\mathbb{F}_2$ and $\mathbb{F}_3$}}\ar@{<-}[ur]\ar@{<-}[ul] & \\
			}
		\]
	\caption{A diagram of weights.}
	\label{diagram_metrics}
\end{figure*}
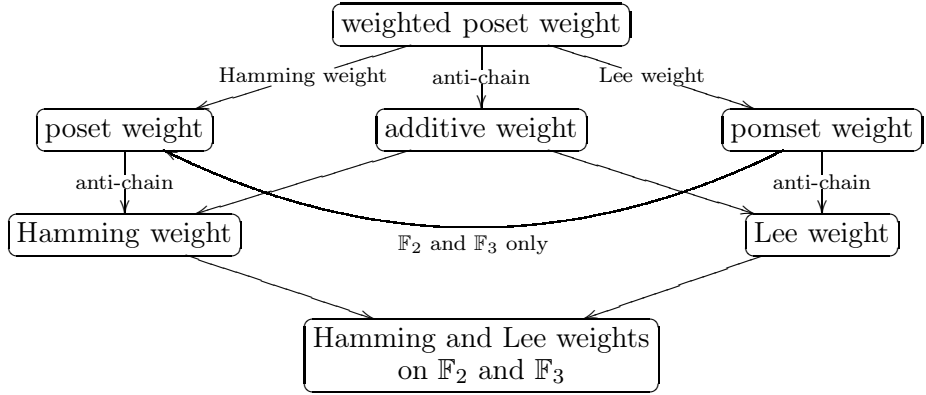


\subsection{Anticodes and Diameter Perfect Codes}\label{sec_anti_diam}

Let $A$ be a subset of the weighted poset metric space $(\mathbb{F}_q^n,d_{(P,w)})$. The \emph{diameter} of $A$ is the maximum $(P,w)$-distance occurring between elements of $A$: $$diam_{d_{(P,w)}}(A) \eqdef \max \{d_{(P,w)}(x,y) : x,y \in A \}.$$ In this case we say that $A$ is an \textit{anticode} with diameter $diam_{d_{(P,w)}}(A)$. Let $$A_{d_{(P,w)}}^*(D) \eqdef \max \{ |A| : diam_{d_{(P,w)}}(A) \leq D \}.$$ An anticode $A$ is called $D$-\emph{optimal} if $|A| = A_{d_{(P,w)}}^*(D)$.

The class of the weighted poset metric spaces $(\mathbb{F}_q^n,d_{(P,w)})$ is distinct from the class of the distance regular graph metric spaces and, thus, Delsarte's Theorem, does not apply. Fortunately, an alternative proof that was given in \cite{Deza_set_antiset} can be slightly modified to work on this case (and for all metrics invariant by translations).

\begin{theorem}\label{theo_Delsart_inv_trans}
	If $A$ and $B$ are subsets of the weighted poset metric space $(\mathbb{F}_q^n,d_{(P,w)})$ such that nonzero $(P,w)$-distances occurring between elements in $A$ do not occur between elements of $B$, then $$|A| \cdot |B| \leq q^n.$$
\end{theorem}

\begin{proof}
	Suppose that $A$ and $B$ are subsets of $\mathbb{F}_q^n$ such that nonzero $(P,w)$-distances occurring between elements in $A$ do not occur between elements of $B$. Let $$A+B \eqdef \{a+b : a \in A, \ b \in B\}.$$ Let us show that $|A+B| = |A| \cdot |B|$. Suppose $a_1 + b_1 = a_2 + b_2$ with $a_1,a_2 \in A$ and $b_1,b_2 \in B$. Since $d_{(P,w)}$ is invariant by translations,
\begin{eqnarray*}
	d_{(P,w)}(a_1,a_2) &=& d_{(P,w)}(a_1 + b_1,a_2 + b_1) \\
                       &=& d_{(P,w)}(a_2 +b_2,a_2 + b_1) \\
                       &=& d_{(P,w)}(b_2,b_1).
\end{eqnarray*}
As nonzero $(P,w)$-distances occurring between elements in $A$ do not occur between elements of $B$, it follows that $d_{(P,w)}(a_1,a_2) = d_{(P,w)}(b_2,b_1) = 0$. Therefore $a_1 = a_2$ and $b_1 = b_2$ and so   $|A+B| = |A| \cdot |B|$. Since $A+B \subseteq \mathbb{F}_q^n$, we conclude that $|A| \cdot |B| = |A+B| \leq |\mathbb{F}_q^n| = q^n$.
\end{proof}

A \emph{code} is any subset $C \subseteq \mathbb{F}_{q}^{n}$ with \emph{minimum distance} $$d_{(P,w)}(C) \eqdef \min\{d_{(P,w)}(x,y): x,y \in C,\ \ x \neq y\}.$$

\begin{corollary}
	For all codes $C \subseteq \mathbb{F}_q^n$ with minimum distance $d_{(P,w)}(C) = D+1$ we have
	\begin{equation}\label{diam.perf.x}
	A_{d_{(P,w)}}^*(D) \cdot |C| \leq q^n.
	\end{equation}
\end{corollary}

A code $C\subseteq \mathbb{F}_q^n$ is called \textit{diameter perfect} if (\ref{diam.perf.x}) holds with equality.

In this work, we will study the diameter perfect codes and optimal anticodes on $(\mathbb{F}_q^n,d_{(P,w)})$ when $P$ is a chain order. We start by presenting the basic results of this weighted poset metric spaces.


\section{Codes and Anticodes on Chain Orders}\label{sec_NRT}

An order $P$ on the finite set $[n]$ is called a \emph{chain} if every two elements are \emph{comparable}, i.e., given $i,j \in \left[  n\right]$ we have that either $i \leq_P j$ or $j \leq_P i$. In this case, the set can be labeled in such a way
that $i_1 \leq_P i_2 \leq_P \ldots \leq_P i_n$. For the simplicity of the notation, in this situation, we will always assume that the chain order $P$ is the usual chain order $$1 \leq 2 \leq \ldots \leq n.$$

For the usual chain order $P = ([n],\leq)$ we have that $| \left\langle i \right\rangle | = i$ for each $i \in [n]$. So, given $0 \neq u = (u_1, \ldots , u_n) \in \mathbb{F}_q^n$,
\begin{equation}\label{peso_NRT}
\varpi_{(P,w)}(u) =  w(u_i) + (i - 1) \cdot M_w,
\end{equation}
where $i = \max\{j : u_j \neq 0\}$.
The metric space $(\mathbb{F}_q^n , d_{(P,w)})$ will be called \textit{weighted chain metric space}.  We remark that the weighted chain metric space is a particular case (a single chain) of the Niederreiter-Rosenbloom-Tsfasman spaces (several chains) if $w$ is the Hamming weight. Originally, the Niederreiter-Rosenbloom-Tsfasman spaces were introduced by Niederreiter in \cite{Niede} and Rosenbloom and Tsfasman in \cite{RT}. The Niederreiter-Rosenbloom-Tsfasman spaces are of special interest since they have several applications, as noted by Rosenbloom and Tsfasman (see \cite{RT}) and Park and Barg (see \cite{Barg}).

In this section, we will away assume that $P = ([n],\leq)$ is the usual chain order and develop several results on codes and anticodes. Besides, we will omit the index $P$ and write just $d_w$ and $\varpi_w$ for the \emph{weighted chain distance} $d_{(P,w)}$ and \emph{weighted chain weight} $\varpi_{(P,w)}$, respectively. In particular, if $w$ is the Hamming weight, we will write just $d$ and $\varpi$, the poset distance and poset weight, respectively.

We will denote by $B_w(u,r)$ the \textit{$r$-ball} with center $u$ and radius $r$ relative to the weighted chain distance $d_w$, i.e.,
\[
B_w(u,r) \eqdef \{v \in \mathbb{F}_q^n : d_w(u,v) \leq r \}.
\]

\begin{remark}
Writing $r = s + i M_w$ with $0 \leq s < M_w$, if $s \neq w(a)$ for all $a \in \mathbb{F}_q$, then $$B_w(u,r) = B_w(u,t + i M_w)$$ for all $s' \leq t \leq s$ where $s'$ is the largest integer such that $s' < s$ and $s' = w(a)$ for some $a \in \mathbb{F}_q$.
\end{remark}

From now on we will assume that $r = s + i M_w$ with $s = w(a)$ for some $a \in \mathbb{F}_q$.

Given $r \in \mathbb{N}$ such that $r \geq 1$ and $w$ be a weight on $\mathbb{F}_q$, let $$[r]_w \eqdef \{t \in [r] : t=w(a) \text{ for some } a \in \mathbb{F}_q\}$$ be the $w$-\textit{interval}. We have that $[r]_w = [r] \cap Im(w)$, where $Im(w)$ is the image of the map $w$.

\begin{proposition}[Size of Ball]\label{lemma_bolas}
	Let $D = S + R \cdot M_w$ be a non-negative integer and $x \in \mathbb{F}_q^n$. If $S > 0$, then
	\begin{eqnarray*}
		|B_w(x,D)| = q^R \cdot (1 + |w^{-1}([S]_w)|).
	\end{eqnarray*}
	If $S = 0$, then $$|B_w(x,D)| = q^R.$$
\end{proposition}

\begin{proof}
	See Appendix \ref{appendix_2}.
\end{proof}

\ \

If $X$ is a subset of $\mathbb{F}_q^n$, the \textit{packing radius} $R_w(X)$ is the largest positive integer number $r$ such that any two $r$-balls centered at distinct elements of $X$ are disjoint. In \cite[Corollary 22]{Panek_Pinheiro} the authors show that
\begin{equation}\label{packing_radius}
	R_w(X) = M_w \cdot (d(C)-1),
\end{equation}
where $d(C)$ is the \emph{minimum distance} of $C$ relative to the poset metric $d$.

We say that a code $C$ is \textit{perfect} if the union of the $r$-balls, $r=R_w(C)$, centered at the elements of $C$, covers $\mathbb{F}_q^n$.

\subsection{MDS Codes}

In this section, we describe all the MDS codes on $(\mathbb{F}_q^n , d_w)$. We shall see later (Section \ref{sec_diam_perf_codes}) that these will be the only diameter perfect codes for many instances of $(\mathbb{F}_q^n , d_w)$. We begin by showing that the minimum distance $d_w(C)$ is determined by $d(C)$.

Given a weight $w$ on $\mathbb{F}_q$, let $m_w \eqdef \min\{w(\alpha): 0 \neq \alpha \in \mathbb{F}_q\}$.

\begin{proposition}\label{lemma_dist_min_geral}
	Let $x_i$ be the $i^{th}$ coordinate of $x \in \mathbb{F}_q^n$ and $C$ be a code on $\mathbb{F}_q^n$. Then
\begin{equation*}
  d_w(C) = S_{w,C} + (d(C) - 1) \cdot M_w,
\end{equation*}
\end{proposition}
where $S_{w,C} \eqdef \min\{w(x_{d(C)} - y_{d(C)}) : x,y \in C, \ x_{d(C)}\neq y_{d(C)}\}$. Therefore, if $C$ is a linear code, then $$d_w(C) = m_w + (d(C) - 1) \cdot M_w.$$

\begin{proof}
	Write $d_w(C) = S + R \cdot M_w$ with $m_w \leq S \leq M_w$. This implies that $d(c,c') \geq R + 1$ for all $c \neq c'\in C$. Since $d(c,c') \geq d(C)$ for all $c \neq c'\in C$ and there are $c,c'\in C$ such that $d(c,c') = d(C)$, we conclude that $R + 1 = d(C)$, i.e., $R = d(C) - 1$. The minimality of $d_w(C)$ implies that $S = \min\{w(x_{d(C)} - y_{d(C)}) : x,y \in C, \ x_{d(C)} \neq y_{d(C)}\}$.
\end{proof}

\ \

A code $C \subseteq \mathbb{F}_q^n$ is said to be \emph{maximum distance separable} (MDS) if its size $|C|$ attains the \emph{Singleton bound} $$|C| \leq q^{n-d(C)+1}.$$ Proceeding with the same argument in the proof of Singleton bound (see \cite{Firer_book}, p. 63), we get: if $C$ is an MDS code on $\mathbb{F}_q^n$, then $$C = \{(x_y,y) : y \in \mathbb{F}_q^{n-d(C)+1}\},$$ where $y \mapsto x_y$ is a map from $\mathbb{F}_q^{n-d(C)+1}$ into $\mathbb{F}_q^{d(C)-1}$. We now notice the following: since  $|B_w(0,R_w(C))|=q^{d(C)-1}$ (see (\ref{packing_radius})),
\begin{eqnarray*}
	|C| = q^{n-d(C)+1} &\Leftrightarrow & q^{d(C)-1} \cdot |C| = q^n \\ 
	                   &\Leftrightarrow & |B_w(0,R_w(C))| \cdot |C| = q^n.
\end{eqnarray*}
In short:

\begin{theorem}\label{theo_MDS_perfect}
	In the weighted chain metric space $(\mathbb{F}_q^n , d_w)$, a code $C$ is MDS if and only if $C$ is perfect. Furthermore, if $C \subseteq \mathbb{F}_q^n$ is MDS (perfect), then $$C = \{(x_y,y) : y \in \mathbb{F}_q^{n-d(C)+1}\},$$ where $y \mapsto x_y$ is a map from $\mathbb{F}_q^{n-d(C)+1}$ into $\mathbb{F}_q^{d(C)-1}$.
\end{theorem}


\subsection{Diameter Perfect Codes}\label{sec_diam_perf_codes}

In this section we will show that all MDS codes on $(\mathbb{F}_q^n , d_w)$ are diameter perfect. We will also present some instances where (\ref{diam.perf.x}) is equivalent to the Singleton bound. Consequently, in these instances, diameter perfect and MDS are equivalent concepts.

Let $A \subseteq \mathbb{F}_q^n$ such that $diam_{d_w}(A) \leq D$. Since $d_w(x,y) \leq D$ for all $x,y \in A$, we have that
\begin{equation}\label{A_contido_B(x,D)}
	A \subseteq B_w(x,D)
\end{equation}
for each $x \in A$. Hence $|A| \leq |B_w(x,D)|$. Therefore,
\begin{equation}\label{A*_menor_B(x,D)}
	A^*_{d_w}(D) \leq |B_w(x,D)|,
\end{equation}
by Lemma \ref{lemma_bolas} it follows that:

\begin{theorem}\label{lemma_bould_A^*}
	Let $D = S + R \cdot M_w$ be a non-negative integer with $0 \leq S < M_w$. If $S>0$, then
\begin{equation}\label{bound_A^*}
  A_{d_w}^*(D) \leq q^R \cdot (1 + |w^{-1}([S]_w)|).
\end{equation}
If $S = 0$, then $$A_{d_w}^*(D) \leq q^{R}.$$
Consequently, $A_{d_w}^*(D) < q^{R+1}$.
\end{theorem}

In Section \ref{section_nonarchimedian} we will show that (\ref{bound_A^*}) holds with equality if and only if the weight $w$ is non-archimedean.

\begin{theorem}\label{theo_1_appendix}
	Let $D = \varpi_w(y)$ for some $y \in \mathbb{F}_q^n$. The following properties are equivalent:
	\begin{enumerate}
		\item $diam_{d_w}(B_w(x,D)) = D$ for all $x \in \mathbb{F}_q^n$;
		\item $B_w(x,D)$ is $D$-optimal for all $x \in \mathbb{F}_q^n$;
		\item If $A$ is $D$-optimal, then $A = B_w(x,D)$ for some $x \in \mathbb{F}_q^n$;
		\item $A_{d_w}^*(D) =  |B_w(0,D)|$.
	\end{enumerate}
\end{theorem}

\begin{proof}
	(1) $\Rightarrow$ (2): Let us suppose that $diam_{d_w}(B_w(x,D)) = D$ for all $x \in \mathbb{F}_q^n$. Since $A \subseteq B_w(x,D)$ for all $x \in A$ whenever $diam_{d_w}(A) \leq D$, we get $|B_{w}(x,D)| = A_{d_w}^*(D)$ for all $x \in \mathbb{F}_q^n$. Hence $B_{w}(x,D)$ is $D$-optimal for all $x \in \mathbb{F}_q^n$.
	
	(2) $\Rightarrow$ (3): Now suppose that $B_{w}(x,D)$ is $D$-optimal for all $x \in \mathbb{F}_q^n$ and let $A$ be an anticode such that $|A| = A_{d_w}^*(D)$, i.e., $A$ is $D$-optimal.  Since $A \subseteq B_w(x,D)$ for all $x \in A$ and $|B_{w}(x,D)| = A_{d_w}^*(D)$, then $|A| = |B_{w}(x,D)|$ for all $x \in A$. Thus $A = B_{w}(x,D)$ for all $x \in A$.
	
	(3) $\Rightarrow$ (4): Suppose that $A$ is $D$-optimal and $A = B_{w}(x,D)$ for some $x \in \mathbb{F}_q^n$. Then $A_{d_w}^*(D) = |A| = |B_{w}(x,D)|$. Since $d_w$ is invariant by translations, $|B_{w}(x,D)| = |B_{w}(0,D)|$. Thus $A_{d_w}^*(D) = |B_{w}(0,D)|$.
	
	(4) $\Rightarrow$ (1): Let us suppose that $A_{d_w}^*(D) =  |B_{w}(0,D)|$. If $A$ is $D$-optimal, then $|A| = |B_{w}(0,D)|$. Since $A \subseteq B_w(x,D)$ for all $x \in A$ and $d_w$ is invariant by translation, $|A| = |B_{w}(0 ,D)| = |B_{w}(x ,D)|$ for all $x \in A$, and hence $A = B_w(x,D)$ for all $x \in A$.  This implies that $B_{w}(x,D)$ is $D$-optimal. Therefore $diam_{d_w}(B_{w}(x,D)) \leq D$. As $D = \varpi_w(y)$ for some $y \in \mathbb{F}_q^n$ and $d_w$ is invariant by translation, $D = d_w(x,x+y)$, which implies that $diam_{d_w}(B_{w}(x,D)) = D$ for all $x \in A$. Thus $diam_{d_w}(B_{w}(x,D)) = D$ for all $x \in \mathbb{F}_q^n$.
\end{proof}

\begin{proposition}\label{corol_B(0,D)_subspace}
	Let $D = R \cdot M_w$ and $x \in \mathbb{F}_q^n$. The following properties holds (and are equivalent):
\begin{enumerate}
	\item $diam_{d_w}(B_w(x,D)) = D$;
    \item $B_w(x,D)$ is $D$-optimal;
    \item $B_w(0,D)$ is an $R$-dimensional subspace of $\mathbb{F}_q^n$.
\end{enumerate}
Consequently, $A_{d_w}^*(D) = q^R$.
\end{proposition}

\begin{proof}
	We have that $$B_w(0,D) = \{(x_1,\ldots,x_R,0,\ldots,0) : x_1,\ldots,x_R \in \mathbb{F}_q\}.$$ Hence $B_w(0,D)$ is an $R$-dimensional subspace of $\mathbb{F}_q^n$. Therefore, $x-y \in B_w(0,D)$ for all $x,y \in B_w(0,D)$ and hence $diam_{d_w}(B_w(0,D)) = D$. Since $d_w$ is invariant by translations, $diam_{d_w}(B_w(x,D)) = D$ for all $x \in \mathbb{F}_q^n$. By Theorem \ref{theo_1_appendix} 
	and Proposition \ref{lemma_bolas} it follows that item 2 and $A_{d_w}^*(D) = q^R$ holds, and hence items 1, 2 and $A_{d_w}^*(D) = q^R$ are equivalent.

Now if $B_w(0,D)$ is $D$-optimal, $diam_{d_w}(B_w(0,D)) \leq D$, and since $D = R \cdot M_w$, we have that $x - \lambda y \in B_w(0,D)$ for all $x,y \in B_w(0,D)$ and $\lambda \in \mathbb{F}_q$, i.e., $B_w(0,D)$ is a subspace of $\mathbb{F}_q^n$. 

Thus the items 1, 2, 3 and $A_{d_w}^*(D) = q^R$ holds and are all equivalent.
\end{proof}

\ \

Let $w$ be a weight on $\mathbb{F}_q$ and $D$ be a non-negative integer. We denote by $\lfloor D \rfloor_w$ the largest weight $w(a)$ such that $w(a) < D$.

As $\lfloor D \rfloor_w = R \cdot M_w$ whenever $D = m_w + R \cdot M_w$, by Proposition \ref{corol_B(0,D)_subspace} it follows that $A_{d_w}^*(\lfloor D \rfloor_w) = q^R$. Hence:

\begin{corollary}\label{corol_22}
	Let $C$ be a code with minimum distance $d_w(C) = m_w + (d(C) - 1) \cdot M_w$. Then
\begin{equation}\label{bound_A*}
	A^*_{d_w}(\lfloor d_w(C) \rfloor_w) \cdot |C| \leq q^n
\end{equation}
is equivalent to the Singleton bound. Therefore, if $C$ is a code with minimum distance $d_w(C) = m_w + (d(C) - 1) \cdot M_w$, then $C$ is diameter perfect if and only if $C$ is MDS.
\end{corollary}
 
Let $w_H$ be the Hamming weight on $\mathbb{F}_q$. Since $d_w(C) = d(C) \cdot M_w$ whenever $w = \lambda w_H$ for some integer $\lambda > 0$:

\begin{corollary}\label{corol_24}
	If $w = \lambda w_H$ for some integer $\lambda > 0$, then $$A^*_{d_w}(\lfloor d_w(C) \rfloor_w) \cdot |C| \leq q^n$$ is equivalent to the Singleton bound. Therefore, a code $C$ is diameter perfect if and only if $C$ is MDS.
\end{corollary}

In \cite{ahl} Ahlswede \textit{et al.} proved that MDS codes in Hamming spaces are diameter perfect. This is also our case:

\begin{theorem}\label{theo_MDS_is_diamter_perfect}
	If $C$ is an MDS code with minimum distance $d_w(C)=D$, then $$A^*_{d_w}(\lfloor D \rfloor_w) \cdot |C| = q^n.$$
\end{theorem}

\begin{proof}
	If $C$ is an MDS code with minimum distance $d_w(C)=D$, then $D = m_w +(d(C)-1) \cdot M_w$. This implies that $\lfloor D \rfloor_w = (d(C)-1) \cdot M_w$, and hence $A_{d_w}^*(\lfloor D \rfloor_w) = q^{d(C)-1}$ (see Proposition \ref{corol_B(0,D)_subspace}). Since $|C| = q^{n-d(C)+1}$, the result follows. 
\end{proof}

\ \

The MDS codes (linear and non-linear) are described in Theorem \ref{theo_MDS_perfect}.

Let $C$ be a code on $\mathbb{F}_q^n$ such that $|C| = q^k$ for some $0 \leq k \leq n$. Suppose $d_w(C) = S + (d(C) - 1) \cdot M_w$ with $m_w \leq S < M_w$ be a distance (see Proposition \ref{lemma_dist_min_geral}).  If $C$ is not an MDS code, then $|C| \leq q^{n-d(C)+1-i}$ for some integer $i \geq 1$, and since $A^*_{d_w}(\lfloor D \rfloor_w) < q^{d(C)}$ (see Theorem \ref{lemma_bould_A^*}), we have $$A^*_{d_w}(\lfloor D \rfloor_w) \cdot |C| < q^{n + 1 - i} \leq q^n.$$ Now, if $|C|$ is an MDS code, then $|C| = q^{n - d(C) + 1}$ and $d_w(C) = m_w + (d(C) - 1) \cdot M_w$. If $D = d_w(C)$, by Proposition \ref{corol_B(0,D)_subspace}, $A^*_{d_w}(\lfloor D \rfloor_w) = q^{d(C) - 1}$. Thus $$A^*_{d_w}(\lfloor D \rfloor_w) \cdot |C| = q^n.$$ In short:

\begin{theorem}\label{theo_26}
A code $C$ of size power of $q$ is diameter perfect if and only if $C$ is MDS. In particular, a linear code $C$ is diameter perfect if and only if $C$ is MDS.
\end{theorem}

As we shall see in Theorem \ref{lemma_diam_S_w-1} and Theorem \ref{theo2_A*}, not always the inequality (\ref{bound_A*}) is equivalent to the Singleton bound.


\subsection{Optimal Anticodes}\label{section_optimalanticodes}

In this section, we determine all the $D$-optimal anticodes. If $D = S + R \cdot M_w$, the idea is to partition the $w$-interval $[M_w]_w$ in ``non-archimedean'' and ``not always non-archimedean'' elements and analyze $S$ (and consequently $D$) considering this partition. More precisely, if $S$ is in the part ``non-archimedean'' then all $D$-optimal anticodes are balls of radius $D$ (Theorem \ref{lemma_diam_S_w-1}), and if $S$ is in part ``not always non-archimedean'' then all $D$-optimal anticodes are proper parts of balls of radius $D$ (Theorem \ref{theo2_A*}).

We say that a weight $w$ on $\mathbb{F}_q^n$ is \emph{non-archimedean} if
\begin{equation*}
	w(x+y) \leq \max\{w(x),w(y)\}
\end{equation*}
for all $x,y \in \mathbb{F}_q^n$. Otherwise, we will say that the weight is \emph{archimedean}.

\begin{example}
	The Lee weight $w_L$ on $\mathbb{Z}_m$ is archimedean if $m \geq 4$: for $x = y =1$, $$w_L(x+y) > \max\{w(x),w(y)\}.$$ The Hamming weight $w_H$ on $\mathbb{Z}_m$ is non-archimedean.
\end{example}

Given an archimedean weight $w$ on $\mathbb{F}_q$, let $m_w \leq S_w <M_w$ be the integer
\begin{align*}
	 S_w\eqdef \min\{\max\{&w(a),w(b)\} : a,b \in \mathbb{F}_q \textrm{ and } \\ & w(a-b) > \max\{w(a),w(b)\}\}.
\end{align*}
If $w$ is non-archimedean, we define $S_w \eqdef M_w$. Notice that $S_w > 0$ for all weights $w$.

\begin{theorem}\label{lemma_diam_S_w-1}
	Let $(\mathbb{F}_q^n,d_w)$ be the weighted chain metric space and $D$ be a non-negative integer. Write $D = S + R \cdot M_w$ with $0 \leq S < M_w$. If $0 \leq S < S_w$, then:
\begin{enumerate}
	\item $diam_{d_w}(B_w(x,D)) = D$;
    \item $B_w(x,D)$ is $D$-optimal for all $x \in \mathbb{F}_q^n$;
    \item If $A$ is $D$-optimal, then $A = B_w(x,D)$ for some $x \in \mathbb{F}_q^n$;
    \item $B_w(0,D)$ is an $R$-dimensional subspace of $\mathbb{F}_q^n$ if and only if $S=0$.
\end{enumerate}
Consequently:
\begin{enumerate}
	\item[5.] If $S=0$, then $A_{d_w}^*(D) = q^R$;
    \item[6.] If $S \neq 0$, then $A_{d_w}^*(D) =  q^R \cdot \left(1 + |w^{-1}([S]_w)| \right)$.
\end{enumerate}
\end{theorem}

\begin{proof}
	Since $w(a-b) \leq \max\{w(a),w(b)\}$ if $\max\{w(a),w(b)\} \leq S_w-1$, $a,b \in \mathbb{F}_q$, and since $S<S_w$, it follows that $d_w(x,y) \leq D$ for all $x,y \in B_w(0,D)$. Hence $diam_{d_w}(B_w(0,D)) \leq D$. For $x \in \mathbb{F}_q^n$ such that $\varpi_{P}(x) = R+1$ and $w(x_{R+1}) = S$ we have $d_w(0,x) = D$. Thus $diam_{d_w}(B_w(0,D)) = D$. As $d_w$ is invariant by translations, we get $$diam_{d_w}(B_w(x,D)) = D$$ for all $x \in \mathbb{F}_q^n$. From this and Theorem \ref{theo_1_appendix}, 
	 together with Proposition \ref{lemma_bolas},  we get the items 2, 3, 5 and 6.

Proposition \ref{corol_B(0,D)_subspace} ensures that $B_w(0,D)$ is an $R$-dimensional subspace if $S=0$. Note that if $S_w = m_w$, then $S = 0$. Assume now $m_w \leq S < S_w$. Put $x = (x_1,\ldots,x_{R-1},a,0,\ldots,0)$. If $B_w(0,D)$ is a linear subspace, then $w(\lambda a) \leq S$ for all $\lambda \in \mathbb{F}_q$ and for all $a \in \mathbb{F}_q$ such that $w(a) \leq S$. Since $S < M_w$ there is $b \in \mathbb{F}_q$ such that $w(b) > S$. For $\lambda = ba^{-1}$, we have $w(b) = w(\lambda a) \leq S$, a contradiction. Thus, if $B_w(0,D)$ is a linear subspace of $\mathbb{F}_q^n$, $S=0$.
\end{proof}

\ \

If $A = B_w(x,D)$ and $S \geq S_w$ then $diam_{d_w}(A)>D$, hence $A$ we will not be $D$-optimal.  We will now describe the subsets of $B_w(x,D)$ that are $D$-optimal anticodes.

For $S \geq S_w$, let $W_w(S)$ be a subset of $\mathbb{F}_q$ of maximal size such that:
\begin{enumerate}
	\item if $a \in W_w(S)$, then $S_w \leq w(a) \leq S$;
	\item if $a \in W_w(S)$ and $0 \leq w(b) \leq S_w - 1$, then $w(a-b) \leq S$;
	\item if $a,b \in W_w(S)$, then $w(a-b) \leq S$.
\end{enumerate}

Let $\mathcal{W}_w(S)$ be the set of all subsets $W_w(S)$.
Given integer non-negative $R$  and $K \in \mathcal{W}_w(S)$, let $$Y_{w,R}(K) \eqdef \{x \in \mathbb{F}_q^n : \varpi(x) = R+1 \text{ and } x_{R+1} \in K\}.$$
In other words, if $x \in Y_{w,R}(K)$, then:
\begin{enumerate}
	\item  $|I_x^P| = R+1$;
	\item $S_w + R \cdot M_w \leq d_w(0,x) \leq S + R \cdot M_w$;
	\item $d_w(x,y) \leq S + R \cdot M_w$ for all $y \in \mathbb{F}_q^n$ such that $d_w(0,y) \leq (S_w - 1) + R \cdot M_w$;
	\item $d_w(x,y) \leq S + R \cdot M_w$ for all $y \in Y_{w,R}(K)$.
\end{enumerate}

\begin{theorem}\label{theo2_A*}
	Let $(\mathbb{F}_q^n,d_w)$ be the weighted chain metric space and $D$ be a non-negative integer. Write $D = S + R \cdot M_w$ with $0 \leq S < M_w$ and let $D' = \lfloor S_w \rfloor_w + R \cdot M_w$. If $S \geq S_w$, then:
\begin{enumerate}
	\item $D' \leq diam_{d_w}((x+Y_{w,R}(K)) \cup B_w(x,D')) \leq D$ for all $x \in \mathbb{F}_q^n$ and for all $K \in \mathcal{W}_w(S)$;
	\item $(x+Y_{w,R}(K)) \cup B_w(x,D')$ is $D$-optimal for all $x \in \mathbb{F}_q^n$ and for all $K \in \mathcal{W}_w(S)$;
	\item If $A$ is $D$-optimal, then $A = (x+Y_{w,R}(K)) \cup B_w(x,D')$ for some $x \in \mathbb{F}_q^n$ and $K \in \mathcal{W}_w(S)$.
\end{enumerate}
Consequently,
 $$A_{d_w}^*(D) =  q^R \cdot (1 + |w^{-1}([S_w - 1]_w)| + |W_w(S)|).$$
\end{theorem}

\begin{proof}
Let $K \in \mathcal{W}_w(S)$. If $x,y \in Y_{w,R}(K)$, then $\varpi_w(x) \leq D$, $\varpi_w(y) \leq D$ and $d_w(x,y) \leq D$. This implies that $Y_{w,R} \subseteq B_w(0,D)$ and $diam_{d_w}(Y_{w,R}) \leq D$.

If $$X_{w,R}(K) \eqdef Y_{w,R}(K) \cup B_w(0,D')$$ (see Figure \ref{Y}), we can see that $diam_{d_w}(X_{w,R}(K)) \leq D$, and since $diam_{d_w}(B_w(0,D'))=D'$, it follows that $D' \leq diam_{d_w}(X_{w,R}(K)) \leq D$.

\begin{figure}[h]
\begin{center}
\begin{tikzpicture}[scale=0.45]

\path [draw=none,fill=gray, fill opacity = 0.05] (3,0) circle (6);
\path [draw=none,fill=white, fill opacity = 1] (3,0) circle (3.5);

\draw[dash dot dot, fill=gray!30] (3,0) circle(2.5);

\fill[black]  (3,0) circle (3pt) node [black, below] {\tiny{$0$}};

\draw[dash dot dot] (3,0) circle(3.5);
\filldraw[fill=gray!30, radius=3pt]
    \foreach \i in {2, 3, 4, 13, 14, 15, 17, 21, 28} {
      (3,0) +(0 - \i * 360/30:3.5) circle[]
    };

\draw[loosely dotted] (3,0) circle(3.916);

\draw[loosely dotted] (3,0) circle(4.332);
\filldraw[fill=gray!30, radius=3pt]
    \foreach \i in {0, 1, 3, 11, 12, 15, 18, 19, 20, 25} {
      (3,0) +(0 - \i * 360/31:4.332) circle[fill=gray!30]
    };

\draw[loosely dotted] (3,0) circle(4.748);
\filldraw[fill=gray!30, radius=3pt]
    \foreach \i in {1, 5, 6, 8, 14, 15, 21, 22, 23, 38, 39} {
      (3,0) +(0 - \i * 360/32:4.748) circle[]
    };

\draw[loosely dotted] (3,0) circle(5.164);

\draw[loosely dotted] (3,0) circle(5.58);
\filldraw[fill=gray!30, radius=3pt]
    \foreach \i in {3, 5, 7, 8, 14, 15, 16, 19, 21, 23, 25, 26, 28, 29} {
      (3,0) +(0 - \i * 360/33:5.58) circle[]
    };

\draw[dash dot dot] (3,0) circle(6);
\filldraw[fill=gray!30, radius=3pt]
    \foreach \i in {4, 7, 8, 9, 10, 11, 14, 23, 24, 27, 28, 29} {
      (3,0) +(0 - \i * 360/34:6) circle[]
    };

\draw[ ->] (3,0) -- (5.5,0)  node [midway,fill=white, fill=gray!30] {\tiny{$D'$}};
\draw[->] (3,0) -- ++(66:6)  node [midway,sloped,fill=white] {\tiny{$D$}};
\draw[ ->] (3,0) -- ++(40:3.5)  node [midway,fill=white,fill=gray!30] {\tiny{$\widetilde{D}$}};

\end{tikzpicture}
\caption{The set $X_{w,R}(K)$, where $\widetilde{D} = S_w + R \cdot M_w$.}
\label{Y}
\end{center}
\end{figure}
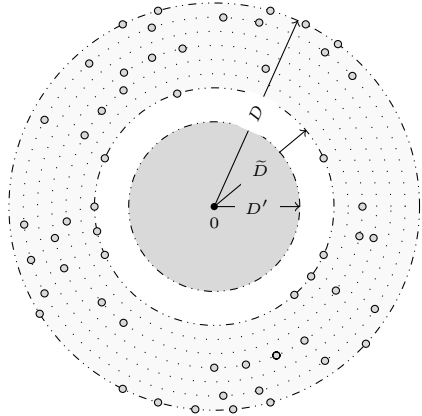

We claim now that if $diam_{d_w}(A) \leq D$, then $A \subseteq x + X_{w,R}(K)$ for some $x \in \mathbb{F}_q^n$ and $K \in \mathcal{W}_w(S)$.
Let $A_a = (-a)+A$ for some $a \in A$.
We have $0 \in A_a$ and $diam_{d_w}(A_a) = diam_{d_w}(A)$. 
Since $A_a \subseteq B_w(0,D)$, $\varpi_w(x) \leq D$ for all $x \in A_a$:
\begin{itemize}
	\item if $\varpi(x) < R+1$, then $x \in B_w(0,D')$;
	\item if $\varpi(x) = R+1$ and $w(x_{R+1}) < S_w$, then $x \in B_w(0,D')$;
	\item if $\varpi(x) = R+1$ and $S_w \leq w(x_{R+1}) \leq S$, since $d_w(x,y) \leq D$ for all $y \in A_a$ such that $\varpi_{P}(y) = R+1$, then $w(x_{R+1} - b) \leq S$ for all $b \in \mathbb{F}_q$ such that $b = y_{R+1}$ for some $y \in A_a$. 
\end{itemize} 
Hence $A_a \subseteq X_{w,R}(K)$ for some $K \in \mathcal{W}_w(S)$, i.e., $A \subseteq x+X_{w,R}(K)$ with $x = a$. Since $diam_{d_w}(x+X_{w,R}(K)) \leq D$, it follows that $A_{d_w}^*(D) = |X_{w,R}(K)|$. As $Y_{w,R}(K) \cap B_w(0,D') = \emptyset$, we conclude that $$A_{d_w}^*(D) =  q^R \cdot (1 + |w^{-1}([S_w - 1]_w)| + |W_w(S)|),$$ and the desire result follows.
\end{proof}

\ \

We conclude this section by showing that (\ref{diam.perf.x}) depends only on $w$ and $d(C)$.

Let $W^w(S)$ be a subset of $\mathbb{F}_q$ of maximal size such that:
\begin{enumerate}
  \item if $a,b \in W^w(S)$ and $a \neq b$, then $w(a-b) \geq S$;
  \item there are $a,b \in W^w(S)$ such that $w(a-b) = S$.
\end{enumerate}

Given an integer $m_w < S <M_w$ such that $S=w(a)$ for some $a \in \mathbb{F}_q$, let
\begin{align*}
	C = \{(0,\ldots,0,&c_R,c_{R+1},\ldots,c_n) : \\ & c_{R+1},\ldots,c_n \in \mathbb{F}_q, \ c_R \in W^w(S) \}.
\end{align*}
We have that $C$ is a code in $\mathbb{F}_q^n$ with minimum distance $d_w(C) = S + (R-1) \cdot M_w$ and size $q^{n-R} \cdot |W^w(S)|$. By Theorem \ref{lemma_diam_S_w-1}, for $D =d_w(C)$ and suppose $S \leq S_w$, it follows that $A_{d_w}^*(\lfloor D \rfloor_w) =  q^{R-1} \cdot \left(1 + |w^{-1}([S-1]_w)| \right)$. Hence, $$A_{d_w}^*(\lfloor D \rfloor_w) \cdot |C| \leq q^n$$ if and only if $$\left(1 + |w^{-1}([S-1]_w)| \right) \cdot |W^w(S)| \leq q.$$ Supposing now $S > S_w$, by Theorem \ref{theo2_A*}, $$A_{d_w}^*(\lfloor D \rfloor_w) =  q^{R-1} \cdot (1 + |w^{-1}([S_w - 1]_w)| + |W_w(S-1)|),$$ and hence, $$A_{d_w}^*(\lfloor D \rfloor_w) \cdot |C| \leq q^n$$ if and only if $$\left(1 + |w^{-1}([S_w - 1]_w)| + |W_w(S-1)| \right) \cdot |W^w(S)| \leq q.$$
The case $S=m_w$ is equivalent to the Singleton bound (Corollary \ref{corol_22}). In short:

\begin{theorem}\label{theo_quando_NRT=Delsate}
	For the weighted chain metric space, we have that (\ref{diam.perf.x}) is equivalent to
\begin{equation}\label{eq1_theo_NRT=delsarte}
	\left(1 + |w^{-1}([S-1]_w)| \right) \cdot |W^w(S)| \leq q
\end{equation}
for all $m_w < S \leq S_w$ and
\begin{align}
	\left(1 + |w^{-1}([S_w - 1]_w)| + |W_w(S-1)| \right) \cdot |W^w(S)| \leq q \label{eq2_theo_NRT=delsarte}
\end{align}
for all $S_w < S < M_w$. Therefore, a code $C$ with minimum distance $d_w(C) = S + R \cdot M_w$ such that $m_w < S < M_w$ is diameter perfect if and only if (\ref{eq1_theo_NRT=delsarte}) or (\ref{eq2_theo_NRT=delsarte}) holds with equality.
\end{theorem}

If $w$ is non-archimedean, then $S_w = M_w$. Therefore:

\begin{corollary}
		For the weighted chain metric space $(\mathbb{F}_q^n,d_w)$ with $d_w$ ultrametric we have that (\ref{diam.perf.x}) is equivalent to
\begin{equation}\label{eq_coro_NRT=delsarte}
	\left(1 + |w^{-1}([S-1]_w)| \right) \cdot |W^w(S)| \leq q
\end{equation}
for all $S > m_w$. Therefore, a code $C$ with minimum distance $d_w(C) = S + R \cdot M_w$ such that $S > m_w$ is diameter perfect if and only if (\ref{eq_coro_NRT=delsarte}) holds with equality.
\end{corollary}


\subsection{Diameter Perfect Codes on $\mathbb{F}_q^n$ with $q$ Prime}\label{sec_diam_per_p_prime}

Now let us suppose that $q =p$ is a prime number. Let $C \subseteq \mathbb{F}_p^n$ be a code with $d_w(C) = S + (d(C) - 1) \cdot M_w$. Since 
\begin{equation*}
	1 + |w^{-1}([S-1]_w)| < p
\end{equation*}
for all $m_w < S < M_w$,
\begin{equation*}
	|W^w(S)| < p
\end{equation*}
for all $m_w < S < M_w$ and
\begin{equation*}
	1 + |w^{-1}([S_w - 1]_w)| + |W_w(S-1)| < p
\end{equation*}
for all $S_w < S < M_w$, by Theorem \ref{theo_quando_NRT=Delsate} if $C$ is diameter perfect, then $S = m_w$. By Corollary \ref{corol_22} and Theorem \ref{theo_MDS_is_diamter_perfect}:
\begin{theorem}\label{theo_p_prime}
	Let $p$ be a prime number. In $(\mathbb{F}_p^n,d_w)$ the only diameter perfect codes are the MDS codes.
\end{theorem}


\section{Non-Archimedean Weights and Ultrametrics}\label{section_nonarchimedian}

We say that a metric $d$ on $\mathbb{F}_q^n$ is an \textit{ultrametric} if
\begin{equation*}
	d(x,y) \leq \max\{d(x,z),d(z,y)\}
\end{equation*}
for all $x,y,z \in \mathbb{F}_q^n$.

\begin{proposition}\label{prop_NRT=ultrametric}
	The weighted chain distance $d_w$ is an ultrametric if and only if $w$ is a non-archimedean weight on $\mathbb{F}_q$.
\end{proposition}

\begin{proof}
	See Appendix \ref{appendix_d_w_ultrametric}.
\end{proof}

\ \

 We know that if $d_{(P,w)}$ is the weighted chain distance with $w$ non-archimedean, then $diam_{d_{(P,w)}}(B_{(P,w)}(x,D)) = D$ for all $x \in \mathbb{F}_q^n$ and for all $D$ (Theorem \ref{lemma_diam_S_w-1} with $S_w = M_w$). Now, we will show that $diam_{d_{(P,w)}} (B_{(P,w)}(x,D)) = D$ for all $x \in \mathbb{F}_q^n$ and for all $D$ only if $d_{(P,w)}$ is the weighted chain distance with $w$ non-archimedean.

\begin{theorem}\label{theo_dimB=D}
	Let $d_{(P,w)}$ be a weighted poset distance. Then $$diam_{d_{(P,w)}}(B_{d_{(P,w)}}(x,D)) = D$$ for all $x \in \mathbb{F}_q^n$ and for all $D$ if and only if $d_{(P,w)}$ is the weighted chain distance $d_w$ and $d_{(P,w)}$ is an ultrametric.
\end{theorem}

\begin{proof}
	Suppose that $d_w$ is an ultrametric. By Proposition \ref{prop_NRT=ultrametric}, $w$ is non-ar\-chi\-me\-dian. This implies that $S_w = M_w$. The ``if'' part follows from Theorem \ref{lemma_diam_S_w-1} with $S_w = M_w$.

Let us suppose now that $P = ([n],\leq_P)$ is not a chain order. Then, there are $i,j \in [n]$ not comparable on $P$. Put $x,y \in \mathbb{F}_q^n$ with $supp(x) = \{i\}$ and $supp(y) = \{j\}$ such that $w(x_i)=w(y_j)$, assuming that $S = \varpi_w(x) \leq \varpi_w(y) = R$, we will have that $x,y \in B_w(0,R)$ and $d_w(x,y) = \varpi_w(x-y) > R$, which implies $diam_{d_w}(B_w(0,R)) > R$. This shows that $P$ is a chain order.

Suppose now that $P$ is the chain order and $d_w$ is not an ultrametric. By Proposition \ref{prop_NRT=ultrametric}, $w$ is ar\-chi\-me\-dian, i.e., there are $a,b \in \mathbb{F}_q$ such that $$w(a-b) > \max\{w(a),w(b)\}.$$ Let $x,y \in \mathbb{F}_q^n$ such that $supp(x)=supp(y)=\{R+1\}$ with $x_{R+1} = a$ and $y_{R+1} = b$. If $S = \max\{w(a),w(b)\}$, then $$\varpi_w(x) = w(a) + R \cdot M_w \leq S + R \cdot M_w$$ and $$\varpi_w(y) = w(b) + R \cdot M_w \leq S + R \cdot M_w.$$ Hence $x,y \in  B_w(0,D)$ with $D = S + R \cdot M_w$. But
\begin{eqnarray*}
	d_w(x,y) &=& \varpi_w(x-y) \\
                   &=& w(a-b) + R \cdot M_w \\
                   &>& S + R \cdot M_w \\
                   &=& D.
\end{eqnarray*}
This show that $diam_{d_w}(B_w(0,D)) > D$, and the ``only if'' part follows.
\end{proof}

\ \

From Theorem \ref{theo_dimB=D}, together with Theorem \ref{theo_1_appendix} 
 and Proposition \ref{lemma_bolas}, it follows that:

\begin{corollary}
	Let $(\mathbb{F}_q^n,d_w)$ be a weighted chain metric space and $D$ be a non-negative integer. Write $D = S + R \cdot M_w$ with $0 \leq S < M_w$. Then $d_w$ is an ultrametric if and only if any of the equivalent properties below holds:
\begin{enumerate}
	\item $diam_{d_w}(B_w(x,D)) = D$ for all $x \in \mathbb{F}_q^n$;
    \item $B_w(x,D)$ is $D$-optimal for all $x \in \mathbb{F}_q^n$;
    \item If $A$ is $D$-optimal, then $A = B_w(x,D)$ for some $x \in \mathbb{F}_q^n$;
	\item $A_{d_w}^*(D) =  q^R$ if $S=0$ and $A_{d_w}^*(D) =  q^R \cdot \left(1 + |w^{-1}([S]_w)| \right)$ if $S \neq 0$.
\end{enumerate}
\end{corollary}



\section{Proof of Proposition \ref{lemma_bolas}}\label{appendix_2}

\begin{proof}
Since $d_w$ is invariant by translations, we have that $$|B_w(x,D)| = |B_w(0,D)|$$ for all $x \in \mathbb{F}_q^n$ and $D \geq 0$. 
Given $x \in \mathbb{F}_q^n$ such that either $\varpi(x) = R+1$ and $w(x_{R+1}) > S$ or $\varpi(x) > R+1$, $$\varpi_w(x) = w(x_{\varpi(x)}) + (\varpi(x)-1) \cdot M_w > D,$$ i.e., $x \notin B_w(0,D)$. Now for each $x \in \mathbb{F}_q^n$ such that either $\varpi(x) = R+1$ and $w(x_{R+1}) \leq S$ or $\varpi(x) < R+1$ we have that $$\varpi_w(x) = w(x_{\varpi(x)}) + (\varpi(x)-1) \cdot M_w \leq D.$$ Hence $x \in B_w(0,D)$ if and only if either $\varpi(x) = R+1$ and $w(x_{R+1}) \leq S$ or $\varpi(x) < R+1$. Thus $$|B_w(0,D)| = q^R \cdot |w^{-1}([S]_w)| + q^R,$$ and the desire result follows.
\end{proof}

\section{Proof of Proposition \ref{prop_NRT=ultrametric}} \label{appendix_d_w_ultrametric}

\begin{proof}
	If $d_w$ is an ultrametric and there are $x,y \in \mathbb{F}_q$ such that $w(x+y) > \max\{w(x),w(y)\}$, taking $u,v \in \mathbb{F}_q^n$ with $supp(u)=supp(v)=\{i\}$ with $u_i=x$ and $v_i=-y$, we have that
\begin{eqnarray*}
	d_w(u,v) &=& \varpi_w(u-v) \\
				   &=& w(u_i-v_i) + (i-1) \cdot M_w \\
			       &>& \max\{w(x),w(y)\} + (i-1) \cdot M_w \\
		           &=& \max\{w(x) + (i-1) \cdot M_w, \\ 
		           & & \hspace{3.0cm} w(y) + (i-1) \cdot M_w\} \\
                   &=& \max\{\varpi_w(u),\varpi_w(v)\} \\
                   &=& \max\{d_w(u,0),d_w(0,v)\},
\end{eqnarray*}
which is a contradiction. Thus $w$ is a non-archimedean weight.

Suppose now that $w$ is a non-archimedean weight. We claim that $\varpi_w$ is a non-archimedean weight on $\mathbb{F}_q^n$: if $x,y \in \mathbb{F}_q^n$ and $i=\max\{j:x_j + y_j \neq 0\}$, then
\begin{eqnarray*}
	\varpi_w(x+y) &=& w(x_i + y_i) + (i-1) \cdot M_w \\
                  & \leq & \max\{w(x_i),w(y_i)\} + (i-1) \cdot M_w \\
                  &=& \max\{w(x_i) + (i-1) \cdot M_w, \\
                  & & \hspace{2.3cm} w(y_i) + (i-1) \cdot M_w\} \\
                  & \leq & \max\{\varpi_w(x),\varpi_w(y)\}.
\end{eqnarray*}
Hence
\begin{eqnarray*}
	d_w(x,y) &=& \varpi_w(x-z-y+z) \\
                & \leq & \max\{\varpi_w(x-z),\varpi_w(z-y)\} \\
                   &=& \max\{d_w(x,z),d_w(z,y)\}
\end{eqnarray*}
for all $x,y \in \mathbb{F}_q^n$. Thus $d_w$ is an ultrametric.
\end{proof}


%
%
%
%
%
%
%

\end{document}